\ifdefined\ORLtwocolumn
  \documentclass[final,5p,times,twocolumn,authoryear]{elsarticle}
\else
  \documentclass[final,11pt,a4paper]{elsarticle}
\fi

\journal{Operations Research Letters}

\usepackage{amsmath,amssymb,amsthm,mathtools}
\usepackage{graphicx,booktabs,threeparttable,multirow,array,float}
\usepackage{microtype,setspace,enumitem}
\usepackage{algorithm,algpseudocode}
\usepackage[font=small,labelfont=bf]{caption}
\usepackage[colorlinks=true,linkcolor=blue,citecolor=blue,urlcolor=blue]{hyperref}
\usepackage[nameinlink,noabbrev]{cleveref}
\makeatletter
\renewcommand{\theHALG@line}{\thealgorithm.\arabic{ALG@line}}
\makeatother

\biboptions{authoryear,round}

\ifdefined\ORLtwocolumn
  \newcommand{\ORLfigurewidth}{\columnwidth}
  
\else
  \newcommand{\ORLfigurewidth}{0.80\textwidth}
  
\fi
\setlist[itemize]{leftmargin=1.35em,itemsep=1pt,topsep=2pt}
\algrenewcommand\algorithmicindent{1em}

\newtheorem{proposition}{Proposition}

\newcommand{\E}{\mathbb{E}}
\newcommand{\Var}{\operatorname{Var}}
\newcommand{\Cov}{\operatorname{Cov}}
\newcommand{\dd}{\,\mathrm{d}}
\newcommand{\ii}{\mathrm{i}}
\newcommand{\calN}{\mathcal{N}}

\begin{document}
\ifdefined\ORLtwocolumn\raggedbottom\fi
\begin{frontmatter}

\title{Efficient simulation schemes for pricing options \\
under the Ornstein--Uhlenbeck driven stochastic volatility model}

\author{Congxin He\corref{cor}}
\ead{cheam@connect.hkust-gz.edu.cn}
\author{Yue Kuen Kwok}
\address{Financial Technology Thrust, \\
Hong Kong University of Science and Technology (Guangzhou)}
\begin{abstract}
We develop an efficient Monte Carlo simulation scheme
for pricing options under the Ornstein--Uhlenbeck driven stochastic volatility model
via the operator splitting approach.
With an ingenious splitting of the governing stochastic differential equations, 
our operator splitting scheme admits analytic solutions in all sub-steps, 
so its implementation is simplified to require simulation 
of a few normal variates. This resolves 
the two typical numerical challenges in other simulation schemes, 
namely, sampling of conditional integrated variance and pathwise inverse
integral transform of characteristic functions. There are three 
pioneering simulation schemes that attempt to overcome the above 
two numerical challenges.  
These include the Hilbert interpolation scheme of \cite{ZengXuJiangKwok2023}, 
Karhunen--Loève expansion scheme of \cite{Choi2025} and moment matching 
scheme based on the Inverse Gaussian distribution of \cite{BrignoneSgarra2026}. 
We performed numerical tests to compare accuracy-speed performance 
of pricing options using our operator splitting scheme with these 
three pioneering schemes. We found that our scheme competes favorably 
well in terms of accuracy-speed tradeoff among all these schemes, 
in particular for pricing path dependent options with a large 
number of monitoring instants. The performance of our scheme can 
be well enhanced by martingale-preserving control variates and
variance reduction via conditioning.  
\end{abstract}


\end{frontmatter}
\ifdefined\ORLtwocolumn
  \setstretch{1.00}
\else
  \setstretch{1.05}
\fi

\section{Introduction}\label{sec:introduction}
\noindent
The Ornstein--Uhlenbeck (OU) driven stochastic volatility models 
belong to a class of financial models that uses an OU process to drive 
the volatility of asset returns \citep{SchobelZhu1999}. The OU 
process is a natural choice for this purpose since it is 
mean-reverting. This is a property widely observed in financial 
market volatility, where volatility tends to cluster around a 
long-term average and regress after extreme movements. The 
Schöbel-Zhu OU model permits correlation between the asset 
price and volatility (the leverage effect), which is crucial 
for capturing the skew and smile patterns observed in options 
markets. A key advantage of the OU-driven models is their analytical 
tractability. Though the OU driven stochastic volatility model 
is not affine, one can still manage to derive various types of 
closed form characteristic functions. Such analytical 
tractability leads to various designs of numerical schemes, 
either via Monte Carlo simulation or discrete Fourier inversion 
algorithms. Recently, \cite{BrignoneSgarra2026} showed that 
the OU driven model is more empirically accurate, confirming 
the general belief that nonaffine models perform better than the 
affine counterparts in calibration performance of traded option 
prices.  

This paper focuses on the design of a new simulation scheme for 
pricing options under the OU driven stochastic volatility model 
via the operator splitting approach. With an ingenious splitting 
of the governing stochastic differential equations, we 
manage to derive closed form solutions in all suboperators in 
the operator splitting procedure. As a result, the operator 
splitting scheme only requires simulation of a few normal variates.  
Next, we present a comprehensive comparison of our proposed 
simulation scheme with three most pioneering simulation schemes. 
These three schemes include the Hilbert interpolation scheme of 
\cite{ZengXuJiangKwok2023}, Karhunen--Loève (KL) expansion 
scheme of \cite{Choi2025} and moment matching scheme based on 
the Inverse Gaussian distribution of \cite{BrignoneSgarra2026}. 
There are other earlier simulation schemes that are not chosen 
in our comparison studies for reasons to be explained below. 
The standard biased Euler discretization scheme is known to 
suffer from the leaking correlation problem since volatility 
may assume negative value. \cite{vanHaastrechtLordPelsser2014} proposed 
a simulation scheme that uses the trapezoidal rule to compute 
the integrated variance. \cite{BrignoneSgarra2026} have shown that 
their moment matching scheme outperforms the trapezoidal rule. 
\cite{LiWu2019} proposed the first exact simulation scheme for 
the OU-driven stochastic volatility model. Their scheme relies on 
their nice derivation of the characteristic function of integrated 
variance conditional on terminal volatility and integrated volatility.
However, the sampling of conditional integrated variance is 
computationally intensive since it involves pathwise Laplace 
inversion of the characteristic function. 
\cite{BrignoneSgarra2026} proposed an innovative moment matching scheme 
that avoids such tedious Laplace inversion procedure, so their scheme 
is shown to outperform the Li-Wu scheme.  
Following a similar approach of Broadie--Kaya exact simulation scheme, 
\cite{ZengXuJiangKwok2023} obtained the characteristic function of 
the log asset price conditional on terminal volatility. 
They then adopted the Hilbert interpolation technique to mitigate 
the time-consuming inverse integral transform. \cite{Choi2025} proposed an exact 
simulation scheme based on the KL expansion for the OU bridge process,
represented as a sine series KL expansion and
with coefficients being independent normal variates. Integrated volatility 
and integrated variance can be analytically expressed as weighted 
sums of independent normal variates, completely avoiding numerical 
Fourier inversion. 

The objectives of our paper are summarized as follows. 
In the first objective, we construct our simulation scheme for pricing 
options under the OU model via the operator splitting approach. 
The scheme is seen to be easily implementable, reliable and compete 
well in accuracy-speed tradeoff. To optimize simulation efficiency, 
our scheme also includes variance reduction enhancements,
including conditioning and martingale-preserving control variates. 
The next objective involves comprehensive comparison studies of our 
operator splitting scheme with three other pioneering simulation 
schemes. In the design of our scheme, we avoid the two numerical 
challenges in earlier simulation schemes, namely, the sampling 
of conditional integrated variance and pathwise inverse integral transform of 
characteristic functions. To achieve better computational efficiency, 
all the four simulation schemes attempt to resolve the above two 
numerical challenges using different approaches. Our numerical tests show
that the operator splitting scheme gives
the best performance in pricing path dependent options with 
a large number of monitoring instants.                 

The remaining sections of the paper are organized as follows. 
In the next section, we present several important theoretical 
properties of the OU driven stochastic volatility model. 
These include the various forms of conditional characteristic 
functions and the sine series KL expansion of the OU bridge. 
In \Cref{sec:schemes}, we present the details of our proposed operator 
splitting scheme. The numerical implementation procedures of the 
other three simulation schemes are also briefly summarized for ease of 
reference and comparison. In particular, we show how these simulation 
schemes attempt to resolve the tediousness of sampling integrated variance 
and inverse integral transform. We also examine the order of 
complexity of the 4 simulation schemes. In \Cref{sec:tests}, we present 
comparison of numerical performance in terms of accuracy-speed 
tradeoff of the 4 simulation schemes in pricing European vanilla 
options, path dependent Asian options, barrier options and corridor 
swaps. We also discuss how to implement the efficiency enhancement 
techniques of conditioning and martingale-preserving control variates.
The paper is ended with 
conclusive remarks in the last section.

\section{Theoretical properties of Ornstein--Uhlenbeck model}\label{sec:theory}
\noindent
In this section, we present several theoretical properties of the 
OU driven stochastic volatility model that are relevant to 
the construction of various simulation schemes for pricing options. 
First, we show that the log asset price has a normal distribution 
conditional on terminal volatility, integrated volatility and 
integrated variance. The joint simulation of terminal volatility 
and integrated volatility can be performed easily via an analytic 
joint distribution. The integrated variance can be computed by 
the pathwise Laplace inversion of the conditional characteristic function 
of the integrated variance derived by \cite{LiWu2019}. 
Next, we present the characteristic function of the log asset 
price conditional on terminal volatility \citep{ZengXuJiangKwok2023}.
This alternative characteristic function is used in two different
type of numerical 
schemes. One use is the simulation of the log asset price 
by an inverse integral transform of the conditional characteristic 
function. Another use is in the construction of the 
recursion quadrature scheme via discrete Fourier inversion 
algorithm. Lastly, we discuss the infinite KL sine series 
expansion with normal variates. One can use the KL expansion to 
compute integrated volatility and integrated variance 
analytically, thus avoiding the tedious Laplace inversion of 
the corresponding characteristic function of conditional
integrated variance.

\subsection{Model formulation and analytic formulas of distribution functions}
\noindent
Following \cite{SchobelZhu1999}, under a filtered probability
space \([\Omega,\mathcal F,\{\mathcal F_t\}_{t\in[0,T]},Q]\), 
the dynamics of log asset price $X_t = \log S_t$ and 
its instantaneous volatility $\sigma_t$ under the 
risk neutral measure \(Q\) are governed by the 
following stochastic differential equations:
\begin{align}
\dd X_t&=\left(r-\frac12\sigma_t^2\right)\dd t
+\sqrt{1-\rho^2}\,\sigma_t\dd W_t^{(1)}
+\rho\sigma_t\dd W_t^{(2)}, \label{eq:model-x}\\
\dd\sigma_t&=\kappa(\theta-\sigma_t)\dd t+\xi\dd W_t^{(2)},\label{eq:model-v}
\end{align}
where $W^{(1)}$ and $W^{(2)}$ are independent Brownian motions. 
Here, $r$ is the risk free interest rate, $\rho$ is the 
instantaneous correlation parameter, $\kappa$ is the 
speed of mean reversion, $\theta$ is the long-term volatility,
$\xi$ is the volatility of volatility parameter. 
Over the time interval $[t,t+h]$, the integrated volatility and 
integrated variance are defined by
\begin{equation}
A_h=\int_t^{t+h}\sigma_s\dd s,
\qquad I_h=\int_t^{t+h}\sigma_s^2\dd s,\label{eq:A-I}
\end{equation}
respectively. Solving Eq.~\eqref{eq:model-v} analytically,
we obtain
\begin{equation}
\sigma_{t+h}=\theta+(\sigma_t-\theta)e^{-\kappa h}
+\xi\sqrt{\frac{1-e^{-2\kappa h}}{2\kappa}}Z,\qquad Z\sim\calN(0,1).\label{eq:ou-transition}
\end{equation}
\cite{LiWu2019} showed that $(\sigma_{t+h},A_h)$ is conditionally bivariate normal. Its conditional mean is
\begin{equation}
\begin{pmatrix}
\theta+(\sigma_t-\theta)e^{-\kappa h}\\[2pt]
\theta h+(\sigma_t-\theta)(1-e^{-\kappa h})/\kappa
\end{pmatrix},\label{eq:joint-mean}
\end{equation}
and the entries of the covariance matrix are
\begin{align}
V_{11}&=\frac{\xi^2}{2\kappa}(1-e^{-2\kappa h}),
&V_{21}=V_{12}&=\frac{\xi^2}{2\kappa^2}(1-e^{-\kappa h})^2,\nonumber\\
V_{22}&=\frac{\xi^2}{\kappa^2}\left[h-\frac{1-e^{-\kappa h}}{\kappa}
-\frac{(1-e^{-\kappa h})^2}{2\kappa}\right].\label{eq:joint-cov}
\end{align}
Direct integration of the following stochastic integral gives
\begin{equation}
\int_t^{t+h}\sigma_s\dd W_s^{(2)}=
\frac{\sigma_{t+h}^2-\sigma_t^2-\xi^2h-2\kappa\theta A_h+2\kappa I_h}{2\xi},
\label{eq:ito-correlation}
\end{equation}
which shows dependence on $A_h$ and $I_h$. Integrating 
Eq.~\eqref{eq:model-x} gives the analytical representation
of log asset return $X_{t+h}-X_t$ in terms of $\sigma_t, \sigma_{t+h}, A_h, I_h$.
More precisely, conditional on
$(\sigma_t,\sigma_{t+h},A_h,I_h)$, the increment $X_{t+h}-X_t$
is a normal distribution:
\begin{equation}
X_{t+h}-X_t\sim\calN(m_h,q_h),\label{eq:conditional-normal}
\end{equation}
where
\begin{align}
m_h={}&\left(r-\frac12\rho\xi\right)h
+\frac{\rho}{2\xi}(\sigma_{t+h}^2-\sigma_t^2)
-\frac{\rho\kappa\theta}{\xi}A_h
+\left(\frac{\rho\kappa}{\xi}-\frac12\right)I_h,\notag \\
q_h={}&(1-\rho^2)I_h.\notag
\end{align}
One can simulate \((\sigma_{t+h}, A_h)\) using the bivariate normal distribution in Eq.~\eqref{eq:joint-mean}.
To simulate \(X_{t+h}\), we can use the normal distribution with mean \(m_h\) and variance \(q_h\) in Eq.~\eqref{eq:conditional-normal}. 
However, both the mean and variance involve the integrated variance \(I_h\).
This reveals the numerical challenge of how to compute $I_h$ efficiently.
\subsection{Li-Wu exact simulation scheme}\label{sec:liwu}
\noindent
Let $I_h^c$ be defined as the integrated variance conditional on
\(\sigma_{t+h}\) and \(\int_t^{t+h} \sigma_s\, \dd s\), where
$$
I_h^c=\int_t^{t+h}\sigma_s^2\dd s\mid\sigma_{t+h},\int_t^{t+h} \sigma_s\, \dd s.
$$
\cite{LiWu2019} managed to derive the Laplace transform of 
$I_h^c$ as follows:
\begin{equation}
\mathcal L_h(\alpha)=
\E_t\!\left[e^{-\alpha \int_t^{t+h}\sigma_s^2\, \dd s}\mid \sigma_{t}, \sigma_{t+h},\int_t^{t+h} \sigma_s\, \dd s\right]
=\frac{f(\sqrt{\kappa^2+2\alpha\xi^2})}{f(\kappa)},\qquad \alpha\ge0,
\label{eq:liwu-laplace}
\end{equation}
where $\E_t$ is the expectation taken conditional on filtration \(\mathcal F_t\) and
\begin{equation}
f(x) = \frac{x^2}{2\pi\sqrt{\eta(x)}}
\exp\!\left(-\frac{g(x)}{2\xi^2\eta(x)}\right).
\end{equation}
For convenience, we write $v=\sigma_t$,
$w=\sigma_{t+h}$, $a=A_h$. The functions
$\eta(x)$ and $g(x)$ are given by
\begin{align*}
\eta(x)&=2-2\cosh xh+xh\sinh xh, \notag\\
g(x)&=2x^2[(v+w)a-vwh]
+x^2[(v^2+w^2)h-2(v+w)a]\cosh xh \notag\\
&\quad+x[x^2a^2-(w-v)^2]\sinh xh. \notag
\end{align*}

\noindent
The Li-Wu exact simulation scheme consists of three steps.
First, we draw \((\sigma_{t+h},A_h)\) from the known bivariate normal distribution.
Second, we compute the conditional integrated variance $I_h^c$ from the Laplace inversion of $\mathcal L_h(\alpha)$.
Third, we draw $X_{t+h}$ from the known normal distribution conditional on 
\(\sigma_t,\sigma_{t+h},A_h\) and \(I_h\).
The pathwise Laplace inversion of $\mathcal L_h(\alpha)$ is time consuming,
thus undermining the computational efficiency of the exact
simulation scheme. \cite{BrignoneSgarra2026} constructed their simulation scheme
based on moment matching via the analytic formula of $\mathcal{L}_{h}(\alpha)$.
They managed to avoid Laplace inversion of $\mathcal L_h(\alpha)$,
the details of which are presented in \Cref{sec:brignone}.

\subsection{Conditional characteristic function of log return}\label{sec:return-transform}
\noindent
\cite{ZengXuJiangKwok2023} derived the following characteristic function
of log return conditional on $\sigma_t$ and $\sigma_{t+h}$ as 
defined by
\begin{equation}
\Psi_h(u;v,w)=\E\!\left[e^{\ii u(X_{t+h}-X_t)}\mid\sigma_t=v,\sigma_{t+h}=w\right].
\label{eq:zengkwok}
\end{equation}
Let $z=\ii u$ and we define the following auxiliary functions:
\begin{align*}
d(z)
&=\sqrt{\kappa^2-(2\kappa\rho-\xi)\xi z-(1-\rho^2)\xi^2z^2},
&\qquad r_1(z)
&=\frac{\kappa-d(z)}{2\xi^2},\\
r_2(z)
&=-\frac{z\rho\kappa\theta}{\xi d(z)}
+\frac{\kappa^2\theta}{\xi^2d(z)}-\frac{\kappa\theta}{\xi^2},
&\qquad \widehat\theta(z)
&=\frac{\kappa\theta+\xi^2r_2(z)}{d(z)}.
\end{align*}
The conditional 
characteristic function of log asset return is given by \citep{ZhangZengKwok2023,ZengXuJiangKwok2023,HuKwok2025}
\begin{align}
\Psi_h(u;v,w)={}&\exp\!\Big(
\left(zr+\xi^2r_1+\tfrac12\xi^2r_2^2+\kappa\theta r_2
-\tfrac12z\xi\rho\right)h \nonumber\\[-2pt]
&\quad+\left(r_1-\frac{z\rho}{2\xi}\right)(v^2-w^2)+r_2(v-w)\Big)
\frac{p_h(w\mid v;d(z),\widehat\theta(z))}
{p_h(w\mid v;\kappa,\theta)}.\label{eq:conditional-cf-explicit}
\end{align}
Here, $p_h(w\mid v;a,b)$ represents the transition density over the 
time interval $h$ under the dynamic equation:
\(\dd V_s=a(b-V_s)\dd s+\xi\dd W_s\).  For real $a>0$ and
$b\in\mathbb R$, we obtain
\begin{align*}
\nu_a^2(h)&=\frac{\xi^2}{2a}\bigl(1-e^{-2ah}\bigr),\\
p_h(w\mid v;a,b)
&=\frac{1}{\sqrt{2\pi\nu_a^2(h)}}
\exp\!\left(-\frac{[w-b-(v-b)e^{-ah}]^2}{2\nu_a^2(h)}\right).
\end{align*}

\cite{ZengXuJiangKwok2023} constructed an alternative exact simulation 
scheme based on the inverse integral transform of the above 
conditional characteristic function. They chose to use the 
Hilbert inversion transform instead of the common Laplace inversion transform. 
To minimize computational efforts, they designed an efficient 
interpolation scheme that mitigates the Hilbert inversion 
calculation for all simulation paths. The details of their 
Hilbert interpolation scheme can be found in \cite{ZengXuJiangKwok2023}.
As a remark, the same conditional characteristic function $\Psi_h(u;v,w)$ is also used 
in the construction of the recursion quadrature scheme 
for pricing path dependent options using
the discrete Fourier inversion transform \citep{ZhangZengKwok2023}.

\subsection{Karhunen--Loève expansion}\label{sec:choi-theory}
\noindent
Analytical tractability of integrated volatility $A_T$ 
and integrated variance $I_T$ is explored via
the KL expansion of the OU bridge process.
The OU bridge is a Gaussian process that admits 
an infinite sine series expansion, thus enabling both $A_T$ and 
$I_T$ to be analytically integrable. 
Following \citet{Choi2025}, we write $\bar\sigma_t=\sigma_t-\theta$ and define
\(\widehat\sigma_T=\bar\sigma_T-\bar\sigma_0e^{-\kappa T}.\)
Conditional on $\widehat\sigma_T$, the OU bridge admits the KL expansion
\begin{equation}
\bar\sigma_t=\bar\sigma_0e^{-\kappa t}
+\widehat\sigma_T\frac{\sinh\kappa t}{\sinh\kappa T}
+\xi\sqrt T\sum_{n=1}^{\infty}a_n
\sin\frac{n\pi t}{T} Z_n,
\quad
a_n=\sqrt{\frac{2}{(\kappa T)^2+(n\pi)^2}},
\label{eq:kl-bridge}
\end{equation}
where $Z_n's$ are independent standard normal variates. 
Analytically, $A_T$ becomes a linear Gaussian series and $I_T$ becomes 
\begin{equation}
I_T=c_0+c(\widehat\sigma_T)^\top Z
+\frac12\xi^2T\sum_{n=1}^{\infty}a_n^2(Z_n^2-1).
\end{equation}
The functionals $A_T$ and $I_T$ are thereby
replaced by linear and quadratic functions of standard normal variates.
In actual implementation, these infinite series 
expansions of \(A_T\) and \(I_T\) have to be 
truncated to a finite number of terms. The simulation 
scheme proposed by \cite{Choi2025} takes advantage of 
the relative ease of computing \(A_T\) and \(I_T\) 
via the KL expansion, the details of which are 
presented in \Cref{sec:choi-scheme}.



\section{Simulation schemes}\label{sec:schemes}
\noindent
In this section, we first present our proposed simulation 
scheme that is derived based on the operator splitting 
approach. Its implementation completely avoids the evaluation of 
integrated variance and inverse integral transform of 
characteristic function. We also summarize the innovative 
procedures used in other simulation schemes that are designed to 
overcome these two numerical challenges in order to enhance computational 
efficiency.

\subsection{Operator splitting scheme}\label{sec:splitting}
\noindent
Operator splitting is a powerful numerical technique for 
solving complex time dependent differential equations by 
decomposing the original operator into simpler suboperators, 
which are then solved successively. The operator splitting 
approach has been used successfully for the lifted Heston and
Barndorff--Nielsen models \citep{HuHeKwokZhang2026}. Related sequential
simulation ideas have also been applied to the SABR model
\citep{ChoiHuKwok2026}. Here, we show how to choose
an ingenious splitting of the stochastic 
differential equations of the OU model into suboperators
such that sampling of \(I_h\) and inverse integral transform 
of characteristic function are avoided.

We choose to split Eqs.~\eqref{eq:model-x} and \eqref{eq:model-v} into the following 
suboperators:
\begin{subequations}\label{eq:split-subsystems}
\ifdefined\ORLtwocolumn
\begin{align}
\dd X_t^{\mathrm{I}}&=\left(r-\tfrac12(\sigma_t^{\mathrm{I}})^2\right)\dd t
+\sqrt{1-\rho^2}\,\sigma_t^{\mathrm{I}}\dd W_t^{(1)},\nonumber\\
\dd\sigma_t^{\mathrm{I}}&=\kappa(\theta-\sigma_t^{\mathrm{I}})\dd t,\label{eq:subsystem-one}\\
\dd X_t^{\mathrm{II}}&=\rho\sigma_t^{\mathrm{II}}\dd W_t^{(2)},\nonumber\\
\dd\sigma_t^{\mathrm{II}}&=\xi\dd W_t^{(2)}.\label{eq:subsystem-two}
\end{align}
\else
\begin{align}
\dd X_t^{\mathrm{I}}&=\left(r-\tfrac12(\sigma_t^{\mathrm{I}})^2\right)\dd t
+\sqrt{1-\rho^2}\,\sigma_t^{\mathrm{I}}\dd W_t^{(1)},
&\dd\sigma_t^{\mathrm{I}}&=\kappa(\theta-\sigma_t^{\mathrm{I}})\dd t,
\label{eq:subsystem-one}\\
\dd X_t^{\mathrm{II}}&=\rho\sigma_t^{\mathrm{II}}\dd W_t^{(2)},
&\dd\sigma_t^{\mathrm{II}}&=\xi\dd W_t^{(2)}.
\label{eq:subsystem-two}
\end{align}
\fi
\end{subequations}
When we integrate Eq.~\eqref{eq:subsystem-one} over $[t, t+h]$,
we obtain
\begin{subequations}\label{eq:split}
\begin{equation}
\ifdefined\ORLtwocolumn
\begin{aligned}
\sigma_{t+h}^{\mathrm{I}}&=D(v,h),\\
X_{t+h}^{\mathrm{I}}&=x+rh-\tfrac12Q(v,h)
+\sqrt{(1-\rho^2)Q(v,h)}Z_1.
\end{aligned}
\else
\begin{aligned}
\sigma_{t+h}^{\mathrm{I}}&=D(v,h),
&\qquad X_{t+h}^{\mathrm{I}}&=x+rh-\tfrac12Q(v,h)
+\sqrt{(1-\rho^2)Q(v,h)}Z_1,
\end{aligned}
\fi
\label{eq:split-a}
\end{equation}
where $x = X_t, v = \sigma_t$ and $Z_1$ is the standard normal variate.
The functions are
\ifdefined\ORLtwocolumn
\begin{align*}
D(v,h)&=\theta+(v-\theta)e^{-\kappa h},\\
Q(v,h)&=\int_0^h D(v,s)^2\dd s =\theta^2h+\frac{2\theta(v-\theta)}{\kappa}(1-e^{-\kappa h})
\nonumber\\
&\qquad+\frac{(v-\theta)^2}{2\kappa}(1-e^{-2\kappa h}).
\end{align*}
\else
\begin{align*}
D(v,h)&=\theta+(v-\theta)e^{-\kappa h},\\
Q(v,h)&=\int_0^h D(v,s)^2\dd s
=\theta^2h+\frac{2\theta(v-\theta)}{\kappa}(1-e^{-\kappa h})
+\frac{(v-\theta)^2}{2\kappa}(1-e^{-2\kappa h}).
\end{align*}
\fi
Integrating Eq.~\eqref{eq:subsystem-two} over $[t, t+h]$ gives
\begin{equation}
\ifdefined\ORLtwocolumn
\begin{aligned}
\sigma_{t+h}^{\mathrm{II}}&=v+\xi\sqrt h Z_2,\\
X_{t+h}^{\mathrm{II}}&=x+\frac{\rho}{2\xi}
\big[(\sigma_{t+h}^{\mathrm{II}})^2-v^2-\xi^2h\big].
\end{aligned}
\else
\begin{aligned}
\sigma_{t+h}^{\mathrm{II}}&=v+\xi\sqrt h Z_2,
&\qquad X_{t+h}^{\mathrm{II}}&=x+\frac{\rho}{2\xi}
\big[(\sigma_{t+h}^{\mathrm{II}})^2-v^2-\xi^2h\big],
\end{aligned}
\fi
\end{equation}
\end{subequations}
where $Z_2$ is a standard normal variate.

The Strang-Marchuk splitting procedure advances 
the solution $(X_t, \sigma_t)$ from \(t\) to \(t+h\) into three 
substeps: half step \(\frac{h}{2}\) using 
suboperator I, full step \(h\) using suboperator 
II, half step \(\frac{h}{2}\) using suboperator I. 
Such splitting can be shown to be second order 
\(O(h^2)\) in global error. The proof of 
convergence of the second order Strang-Marchuk 
splitting can be established following a 
similar proof in \cite{HuHeKwokZhang2026}.

By the Strang-Marchuk operator splitting 
procedure, we obtain the numerical implementation 
of the operator splitting scheme as follows:
\begin{equation}
\begin{aligned}
(\tilde X,\tilde\sigma)
&=
\left(
X^{I}(X_t,\sigma_t,\frac{h}{2}),
\sigma^{I}(\sigma_t,\frac{h}{2})
\right),\\
(\hat X,\hat\sigma)
&=
\left(
X^{II}(\tilde X,\tilde\sigma,h),
\sigma^{II}(\tilde\sigma,h)
\right),\\
(X_{t+h}^{\mathrm{split}},\sigma_{t+h}^{\mathrm{split}})
&=
\left(
X^{I}(\hat X,\hat\sigma,\frac{h}{2}),
\sigma^{I}(\hat\sigma,\frac{h}{2})
\right).
\end{aligned}
\label{eq:strang-splitting}
\end{equation}
By combining the solutions of \(X^I,\sigma^I,X^{II}\) and \(\sigma^{II}\) in Eq.~(\ref{eq:split}a,b) 
and the above splitting procedure in Eq.~\eqref{eq:strang-splitting}, we can derive the explicit formulation of \((X_{t+h}^{\mathrm{split}},\sigma_{t+h}^{\mathrm{split}})\) 
as stated in \Cref{prop:splitting-update}.
\begin{proposition}\label{prop:splitting-update}
The Strang–Marchuk operator splitting scheme advances the solution from \((X_t,\sigma_t)\) to \((X_{t+h}^{split},\sigma_{t+h}^{split})\)
as follows:
\begin{equation}
\begin{split}
X_{t+h}^{\rm split}={}&X_t+rh-\frac12I_t^{\rm split,h}
+\frac{\rho}{2\xi}
\{\widehat\sigma^2-\widetilde\sigma^2-\xi^2h\}
+\sqrt{(1-\rho^2)I_t^{\rm split,h}}Z_1,\\
\sigma_{t+h}^{\rm split}={}&D(\widehat\sigma,h/2),
\end{split}
\label{eq:splitting-update}
\end{equation}
where $Z_1$ and $Z_2$ are independent standard normal variates, and 
\begin{align*}
\widetilde\sigma=D(\sigma_t,h/2),
\quad \widehat\sigma=\widetilde\sigma+\xi\sqrt h Z_2,
\quad I_{t}^{\rm split,h}=Q(\sigma_t,h/2)+Q(\widehat\sigma,h/2).
\end{align*}
\end{proposition}

\begin{proof}
Write $Q_1=Q(\sigma_t,h/2)$ and $Q_2=Q(\widehat\sigma,h/2)$.
The intermediate values 
$\widetilde\sigma$, $\widehat\sigma$, $\widetilde X$, and $\widehat X$ 
involved in 
the two half-steps of suboperator I and the full step of
suboperator II are found to be 
\ifdefined\ORLtwocolumn
\begin{align*}
\widetilde\sigma&=D(\sigma_t,h/2),\\
\widetilde X&=X_t+\frac{rh}{2}-\frac{Q_1}{2}
 +\sqrt{(1-\rho^2)Q_1}\,\varepsilon_1,\\
\widehat\sigma&=\widetilde\sigma+\xi\sqrt h\,Z_2,\\
\widehat X&=\widetilde X+\frac{\rho}{2\xi}
 (\widehat\sigma^2-\widetilde\sigma^2-\xi^2h),\\
\sigma_{t+h}^{\rm split}&=D(\widehat\sigma,h/2),\\
X_{t+h}^{\rm split}&=\widehat X+\frac{rh}{2}-\frac{Q_2}{2}
 +\sqrt{(1-\rho^2)Q_2}\,\varepsilon_2.
\end{align*}
\else
\begin{align*}
\widetilde\sigma&=D(\sigma_t,h/2),
&\widetilde X&=X_t+\frac{rh}{2}-\frac{Q_1}{2}
 +\sqrt{(1-\rho^2)Q_1}\,\varepsilon_1,\\
\widehat\sigma&=\widetilde\sigma+\xi\sqrt h\,Z_2,
&\widehat X&=\widetilde X+\frac{\rho}{2\xi}
 (\widehat\sigma^2-\widetilde\sigma^2-\xi^2h),\\
\sigma_{t+h}^{\rm split}&=D(\widehat\sigma,h/2),
&X_{t+h}^{\rm split}&=\widehat X+\frac{rh}{2}-\frac{Q_2}{2}
 +\sqrt{(1-\rho^2)Q_2}\,\varepsilon_2,
\end{align*}
\fi
where $\varepsilon_1$, $\varepsilon_2$ and $Z_2$ are independent standard
normal variables. Conditional on $Z_2$, the two independent price normal
observe
\[
\sqrt{Q_1}\,\varepsilon_1+\sqrt{Q_2}\,\varepsilon_2
\overset{d}{=}\sqrt{Q_1+Q_2}\,Z_1,
\]
with $Z_1$ being independent of $Z_2$. Combining these intermediate
values and together with $I_t^{\rm split,h}=Q_1+Q_2$ gives
Eq.~\eqref{eq:splitting-update}.
\end{proof}

Eq.~\eqref{eq:splitting-update} involves two normal
variates and elementary functions only. 
For pricing path dependent options with multiple monitoring instants, 
we may divide each time interval over successive monitoring 
instants by \(q\) substeps. Assuming uniform time intervals, 
we have \(qh=T/M\), where \(T\) is the time horizon of the option 
and \(M\) is the number of monitoring instants. For pricing a 
European option, where \(M=1\), we may choose \(q\) to be 8 in order 
to achieve better accuracy since the operator splitting scheme is 
$\mathcal{O}(h^2)$ in global error. For pricing a path dependent 
option with $M=12$ or above, it suffices to choose \(q=1\) for 
sufficient accuracy.




\subsection[Low-bias simulation scheme]{Low-bias simulation scheme via moment matching}\label{sec:brignone}
\noindent
To overcome the numerical bottleneck in the Laplace inversion of 
\(\mathcal{L}_h(\alpha)\) in the Li-Wu exact simulation scheme, 
\cite{BrignoneSgarra2026} proposed to use the inverse Gaussian 
(IG) distribution in the approximate evaluation of conditional 
integrated variance \(I_h^c\) via moment matching method. This is 
done by matching the first two conditional moments of the IG 
distribution with those implied by \(\mathcal{L}_h{(\alpha)}\).

More precisely, conditional on
$(\sigma_t,\sigma_{t+h},\int_t^{t+h}\sigma_s \, \dd s)$, the mean and variance of $I_h^c$ are obtained by
differentiating the logarithm of $\mathcal{L}_h(\alpha)$:
\begin{equation}
\mu_I=-\left.\frac{\partial}{\partial\alpha}
\log\mathcal L_h(\alpha)\right|_{\alpha=0},\qquad
v_I=\left.\frac{\partial^2}{\partial\alpha^2}
\log\mathcal L_h(\alpha)\right|_{\alpha=0}.
\label{eq:ig-moments}
\end{equation}
Closed-form expressions for these derivatives are given in
\cite{BrignoneSgarra2026}. Given the moments $\mu_I$ and $v_I$, we approximate 
$I_h^c$ by a two-moment matched IG distribution with mean $\nu$ and shape parameter $\lambda$.
Matching
the two moments gives
$\nu = \mu_I$ and
$\lambda = \mu_I^3/v_I$. The density function of the approximating IG 
distribution is given by
\begin{equation}
f_{IG}(y)=\left(\frac{\lambda}{2\pi y^3}\right)^{1/2}
\exp\!\left(-\frac{\lambda(y-\mu_I)^2}{2\mu_I^2y}\right),
\qquad y>0.
\label{eq:IG_distribution}
\end{equation}
The simulation scheme of \cite{BrignoneSgarra2026} based on moment matching mimics
closely that of the Li--Wu exact simulation scheme except that sampling of conditional
integrated variance \(I_h^c\) via Laplace inversion of \(\mathcal{L}_h(\alpha)\) is replaced by sampling 
of $\mathcal{L}_h^{\alpha}$ using the approximating IG distribution in Eq.~\eqref{eq:IG_distribution}. 

\subsection[Simulation based on the KL expansion]{Simulation scheme based on the Karhunen-Loève expansion}\label{sec:choi-scheme}
\noindent
We have seen in \Cref{sec:choi-theory} that the path integrals of 
integrated volatility and integrated variance can be expressed as 
infinite linear and quadratic series involving normal variates. In 
actual implementation, one has to truncate these infinite series as 
finite sums. The numerical challenge is to estimate the magnitudes 
of the omitted terms in truncation.

We present a brief description of the numerical procedure used by
\cite{Choi2025}. Consider a generic time interval $[t, t+h]$ of 
length $h$. We write $\bar{\sigma}_t=\sigma_t-\theta$, $a_n=\sqrt{\frac{2}{\kappa^2 h^2+n^2\pi^2}}$ 
and $\phi(x)=\frac{1-e^{-x}}{x}$. Let $Z_0$ be a standard
normal variate. We draw $\widehat\sigma_h=\xi\sqrt{h\phi(2\kappa h)}Z_0$ and 
$\bar\sigma_h=\bar{\sigma}_t e^{-\kappa h}+\widehat\sigma_h$.
We define $\bar{U}_h = \frac1h \int_t^{t+h} (\sigma_s-\theta) \, \dd s$ and 
$\bar{V}_h = \frac1h \int_t^{t+h} (\sigma_s-\theta)^2 \, \dd s$. Choi chose to retain an even
number $L$ terms in the two infinite series. The omitted linear terms are jointly Gaussian and can be
sampled exactly from their covariance matrix.  The omitted quadratic terms 
do not fit into any distribution. As an approximation, Choi used the
moment matching method to fit with a convenient choice of distribution. Conditional 
on $\widehat\sigma_h$, the infinite series of center adjusted integrated volatility 
$\bar{U}_h$ and integrated variance $\bar{V}_h$ are split into finite sums and remainder
terms as follows:
\ifdefined\ORLtwocolumn\begingroup\scriptsize\fi
\begin{align*}
\bar U_h={}&\left(\bar{\sigma}_t+\frac{\widehat\sigma_h}{1+e^{-\kappa h}}\right)
\phi(\kappa h)+2\xi\sqrt h
\left(\sum_{\substack{n\le L\\n\text{ odd}}}
\frac{a_n}{n\pi}Z_n+G_L\right),\\
\bar V_h={}&\E(\bar V_h\mid\widehat\sigma_h)
+\xi\sqrt h\left[\bar{\sigma}_t\left(\sum_{n\le L}n\pi a_n^3Z_n+P_L+Q_L\right)
\right.\\[-2pt]
&\left.\qquad+\bar\sigma_h\left(\sum_{n\le L}(-1)^{n-1}n\pi a_n^3Z_n
+P_L-Q_L\right)\right]
+\frac{\xi^2h}{2}\left(\sum_{n\le L}a_n^2(Z_n^2-1)+R_L\right),
\end{align*}
where
\begin{align*}
\E(\bar V_h\mid\widehat\sigma_h)={}&\bar{\sigma}_t^2\phi(2\kappa h)
+\widehat\sigma_h^2\frac{\sinh2\kappa h-2\kappa h}
{4\kappa h\sinh^2\kappa h}
+\frac{\xi^2}{2\kappa}\left(\coth\kappa h-\frac1{\kappa h}\right)\\
&+\frac{b_0\widehat\sigma_he^{-\kappa h}}{\kappa h}
\left(\frac1{\phi(2\kappa h)}-1\right).
\end{align*}
\ifdefined\ORLtwocolumn\endgroup\fi
The details of calculating the remainders $G_L$, $P_L$, $Q_L$, and $R_L$ can be found 
in \cite{Choi2025}.

\subsection{Computational complexity}\label{sec:complexity}
\noindent
As a simplified approach of assessing computational efficiency, 
we may compare the order of complexity of different simulation 
schemes by counting the required number of sampling of random 
realization from a normal distribution and a uniform distribution. 
This simplified approach neglects related computational steps of 
calculating various functions embedded in the simulation algorithm, 
which may be more time consuming than sampling the normal / uniform 
variates. Let us consider pricing a path dependent option with 
\(M\) monitoring instants, assuming uniform spacing of monitoring 
instants for simplicity. Let \(N\) be the number of simulation 
paths used in the option pricing calculation.

In our proposed operator splitting scheme, we may use \(q\) time steps 
between successive monitoring instants. Since each time step 
requires sampling of 2 normal variates, so the total number of 
sampling of normal variates is \(2qNM\). For a path dependent 
option with monthly monitoring frequency, the choice of \(q=1\) 
serves well to give sufficiently low bias. Since it is common to 
have weekly or even daily monitoring frequency, so the choice 
of \(q=1\) is appropriate for most path dependent option pricing 
calculations. In the case of European options, where \(M=1\), 
we may choose the step size \(h\) to be around \(0.1\). 
Together with the inherent simplicity in the algorithm that 
involves minimal related calculations, the operator 
splitting scheme is seen to be highly efficient for pricing 
path dependent options. It remains to be efficient for pricing 
a European option with maturity shorter than one year.

For the Brignone-Sgarra simulation scheme, it is seen that 
each simulation path requires sampling four normal variates 
and one uniform variate. Therefore, pricing of a path dependent 
option requires sampling of \(4NM\) normal variates and \(NM\) 
uniform variates. However, the moment matching algorithm involves 
many more related computational procedures than those in the 
operator splitting scheme, so it is less competitive among 
the two simulation schemes. For pricing European options, 
both schemes share roughly the same level of efficiency 
when \(q\) is around 8 in the operator splitting scheme.

The Choi simulation scheme is least competitive among the 
three simulation schemes. Let \(L\) be the number of 
terms retained in the finite sums of the linear and quadratic 
terms in the algorithm. It is seen that the Choi simulation 
scheme requires sampling of \((6+L)NM\) normal variates, 
plus higher level of tediousness in related computational procedures 
in the algorithm.

Lastly, it is hard to assess the computational complexity in 
terms of sampling of normal / uniform variates for 
the Hilbert interpolation scheme. The scheme involves  
tedious computation in the Hilbert inversion in a simulation run 
and the construction of the grid values in the interpolation table. 
The implementation procedures are most complicated among 
the four simulation schemes. In \Cref{sec:tests}, our numerical 
tests do show that the Hilbert interpolation scheme still 
exhibits accuracy-speed tradeoff close to those of the 
Brignone-Sgarra scheme and Choi scheme for pricing path 
dependent options.
\section{Numerical tests}\label{sec:tests}
\noindent
In this section, we report the numerical tests that were performed to compare accuracy-speed tradeoffs among the four simulation schemes: (i) operator splitting scheme (labeled “splitting”), (ii) Brignone-Sgarra’s scheme of moment matching with inverse Gaussian distribution (labeled “IG”), (iii) Choi’s scheme using the Karhunen-Loève expansion (labeled “KL”), and (iv) Zeng et al.’s Hilbert interpolation scheme (labeled “Hilbert”) for pricing European vanilla call option, Asian call option, down-and-in barrier call option, and corridor variance swap.
In all numerical tests, we employed two standard variance reduction techniques: conditional Monte Carlo simulation and martingale-preserving control variate. The implementation of these two techniques may differ from one simulation scheme to another one as well as from one option product to another one. As shown in our numerical tests, the reduction in variance can be quite substantial.
Unless otherwise stated, the model parameter values of the options used in our simulation studies follow those in \cite{LiWu2019} and \cite{Choi2025}:
\(
S_0=100, K=100, T=1, \sigma_0=0.2,
\theta=0.2, \kappa=4, \xi=0.1, \rho=-0.7, r=0.09531.
\)

First, we state explicitly how we calculate root-mean-square error (labelled “RMSE”) in our simulation runs. Let \(N\) be the number of simulation paths used in a simulation run of 
calculating an estimate. The simulation run of using \(N\) simulation paths is replicated \(R\) times. In our numerical tests, we used
\(
N=1{,}000,\ 2{,}000,\ 5{,}000,\ 10{,}000,\ 100{,}000
\) and \(R=200\). For a given choice of \(N\), let \(\widehat V_{N,j}\) be the estimate in the $j^{\text{th}}$ replication, \(j=1,2,\ldots,R\). The mean of the \(R\) replications of choosing \(N\) simulation paths is
\(
\frac{1}{R}\sum_{j=1}^{R}\widehat V_{N,j}.
\) Let \(V_{\rm ref}\) be the reference true value of the option. We then have
\begin{align}
\widehat{\operatorname{Bias}}_N
 &=\overline V_N-V_{\rm ref},\quad 
\widehat{\operatorname{SD}}_N=\sqrt{\frac{1}{R}\sum_{j=1}^{R}
 (\widehat V_{N,j}-\overline V_N)^2},\nonumber\\
\widehat{\operatorname{RMSE}}_N
 &=\sqrt{\frac1R\sum_{j=1}^{R}
 (\widehat V_{N,j}-V_{\rm ref})^2}
=\sqrt{\widehat{\operatorname{Bias}}_N^2
+\widehat{\operatorname{SD}}_N^2}.
\label{eq:rmse}
\end{align}

In assessing accuracy-speed tradeoff, we use RMSE as the proxy for accuracy and CPU as the proxy for speed. The CPU is taken to be the average of the CPU required for each estimate calculation over \(R\) replications. For simplicity, we neglect the CPU consumed in preliminary calculations, benchmark construction and other one-time computational work. For example, the one-time construction of the interpolation table in the Hilbert interpolation scheme is not included. As a remark, the Hilbert interpolation table was constructed using 2048 Fourier terms with frequency spacing 0.1.

In brief, conditional Monte Carlo method is a powerful variance reduction technique that replaces random function evaluations with their conditional expectations. The core idea is to integrate out part of the randomness analytically, thereby reducing estimator variance without introducing bias. For the variance reduction technique of martingale preserving control variate, the core idea is to construct control variates that are martingales, ensuring that the estimator remains unbiased while reducing variance. In our later exposition, we illustrate how to use these two variance reduction techniques in pricing various types of options using the four different simulation schemes listed above.
\subsection{European vanilla call option}\label{sec:european}
\noindent
The terminal payoff of a European vanilla call option is $e^{-rT}(S_T-K)^+$, where $S_T$ is 
the terminal asset price and $K$ is the 
strike price. Recall that $X_T = \ln S_T$.
At the $i^{\text{th}}$ simulation path, the three 
simulation schemes of splitting, IG,
and KL first simulate the volatility-related quantities that determine the
conditional mean $Y_i$ and conditional variance $V_i$ of $X_T$.  Given these two
quantities, the only remaining randomness in $X_T$ is the independent Brownian term
 $W^{(1)}$. As deduced from Eq.~\eqref{eq:conditional-normal}, the $i^{\text{th}}$
 simulated value of $X_T$ is given by
\[
X_{T, i}=Y_i+\sqrt{V_i}Z_i,\qquad
Z_i\sim\mathcal N(0,1), \qquad i = 1,2,\ldots,N
\]
where the standard Brownian variate $Z_i$ is independent of the simulated volatility-related quantities.
Let $H_i=e^{-rT}(e^{X_{T, i}}-K)^+$ denote the $i^{\text{th}}$ simulated terminal payoff that would be obtained with
sampling of $Z_i$ included.  Variance reduction via conditioning replaces $H_i$ by the conditional expectation
$\E[H_i\mid Y_i,V_i]$.  The tower property leaves the option value unchanged,
while the conditional-variance identity
\[
\Var(H_i)=\E\!\left[\Var(H_i\mid Y_i,V_i)\right]
          +\Var\!\left(\E[H_i\mid Y_i,V_i]\right)
\]
shows that the conditional estimator cannot have larger variance.  This conditioning method
is particularly convenient here because the corresponding conditional expectation is
available in closed form. Indeed, conditional on $(Y_i,V_i)$, the asset price has the same
law for the terminal value as the Black--Scholes model with total variance $V_i$.
Here, $V_i=(1-\rho^2)I_T^{\rm split}$ for splitting and
$V_i=(1-\rho^2)I_T$ for IG and KL.  With conditional forward price
$F_i=\exp(Y_i+V_i/2)$, integrating out $Z_i$ gives the simulated 
option value $P_i$:
\begin{equation}
P_i=e^{-rT}\operatorname{BSCall}(F_i,V_i;K),
\label{eq:euro-cmc}
\end{equation}
where BSCall is the analytic Black--Scholes call price formula:
\ifdefined\ORLtwocolumn
\begin{align*}
\operatorname{BSCall}(F,V;K)&=F\Phi(d_1)-K\Phi(d_2),\\
d_1&=\frac{\log(F/K)+V/2}{\sqrt V},
\qquad d_2=d_1-\sqrt V.
\end{align*}
\else
\begin{align*}
\operatorname{BSCall}(F,V;K)&=F\Phi(d_1)-K\Phi(d_2),
\quad d_1=\frac{\log(F/K)+V/2}{\sqrt V},
\quad d_2=d_1-\sqrt V.
\end{align*}
\fi
Here, $\Phi$ denotes the cumulative distribution function of a standard normal
variate.
Conditioning removes the residual randomness of $Z_i$. Before the resulting estimates are
averaged, we further apply the martingale correction of
\citet{Choi2025} to obtain the conditional forwards:
\begin{equation}
F_i\longmapsto\omega_NF_i,\qquad
\omega_N=\frac{S_0e^{rT}}{N^{-1}\sum_{i=1}^N F_i}.
\label{eq:multiplicative-correction}
\end{equation}
This enforces the martingale condition $\frac{\omega_N}{N}\sum_{i=1}^N F_i=e^{rT}S_0$ within each replication.
The variance reduction involves two steps:
conditioning removes conditional randomness, and imposition of the 
subsequent martingale correction.

The Hilbert scheme returns $X_{T,i}$ rather than the pair $(Y_i,V_i)$ in each
simulation path.  The above Brownian integration and forward rescaling
are not applicable for this simulation scheme.  We instead reduce its sampling
variance by an additive control variate.  A control variate $C_i$ is a simulated
quantity that is correlated with the simulated terminal option payoff $P_i$ and has a known mean.  For
the Hilbert scheme, we let
\[
P_i=e^{-rT}(S_{T,i}-K)^+, \quad
C_i=e^{-rT}S_{T,i}.
\]
The discounted asset price is a martingale, so $c_0=\E[C_i]=S_0$.  A high simulated
terminal asset price produces both a high call payoff and a high value of
$C_i$.  Subtracting a multiple of the deviation $C_i-c_0$ therefore offsets
part of the common pathwise fluctuation.  This gives the additive control
estimator
\begin{equation}
\widehat V_{\rm cv}=\frac1N\sum_{i=1}^N
[P_i-\widehat\beta(C_i-c_0)],\qquad
\widehat\beta=\frac{\widehat{\Cov}(P,C)}{\widehat{\Var}(C)}.
\label{eq:additive-control}
\end{equation}
For any fixed variance reduction coefficient $\beta$, the correction $\beta(C_i-c_0)$ has mean
zero and hence does not change the target option value.  
We used the sample estimates of $\widehat{\Cov}(P,C)$ and $\widehat{\Var}(C)$
computed within each replication to estimate $\widehat\beta$.

Figure~\ref{fig:european} compares the accuracy-speed tradeoff of the four
simulation schemes for pricing the European vanilla call option 
with parameter values specified in the above.
The reference true value of the vanilla call is $13.214920018$, which was
computed by high precision
Fourier inversion of the Riccati characteristic function derived by
\citet{SchobelZhu1999}.
\Cref{fig:european} shows the plots of RMSE versus CPU time (seconds) for the 4 simulation schemes: splitting, IG, KL and Hilbert for pricing the European vanilla option. Each plot connects 5 points, corresponding to the 5 choices of number of simulation paths \(N\) listed in the above. With a similar amount of CPU consumed, the scheme with lower RMSE is preferred in accuracy-speed tradeoff.
For pricing European vanilla call with one-year maturity, we choose \(q=4\) and \(q=8\) in splitting, where \(q\) is the number of time steps within the one-year life of the vanilla call. This is because we need to choose \(h\) to be sufficiently small in order to reduce the bias in the operator splitting procedure. The choice of a larger value of \(q\) makes splitting less competitive in pricing vanilla options. When we consider pricing of path dependent options, the choice of \(q=1\) suffices in splitting when the number of monitoring instants \(M\) is more than 8. This would make splitting to be highly competitive (see discussion in later subsection).
The plots in \Cref{fig:european} reveal that the 3 simulation schemes: splitting (\(q=8\)), IG and KL show a similar level of performance in accuracy-speed tradeoff, especially at high value of number of simulation paths. For splitting (\(q=4\)), RMSE decreases at slower rate with increasing \(N\) since the bias in the operator splitting procedure persists at a certain level. Hilbert is seen to be the least competitive in all choices of \(N\).
\begin{figure}[H]
\centering
\ifdefined\ORLtwocolumn
\includegraphics[width=0.92\ORLfigurewidth]{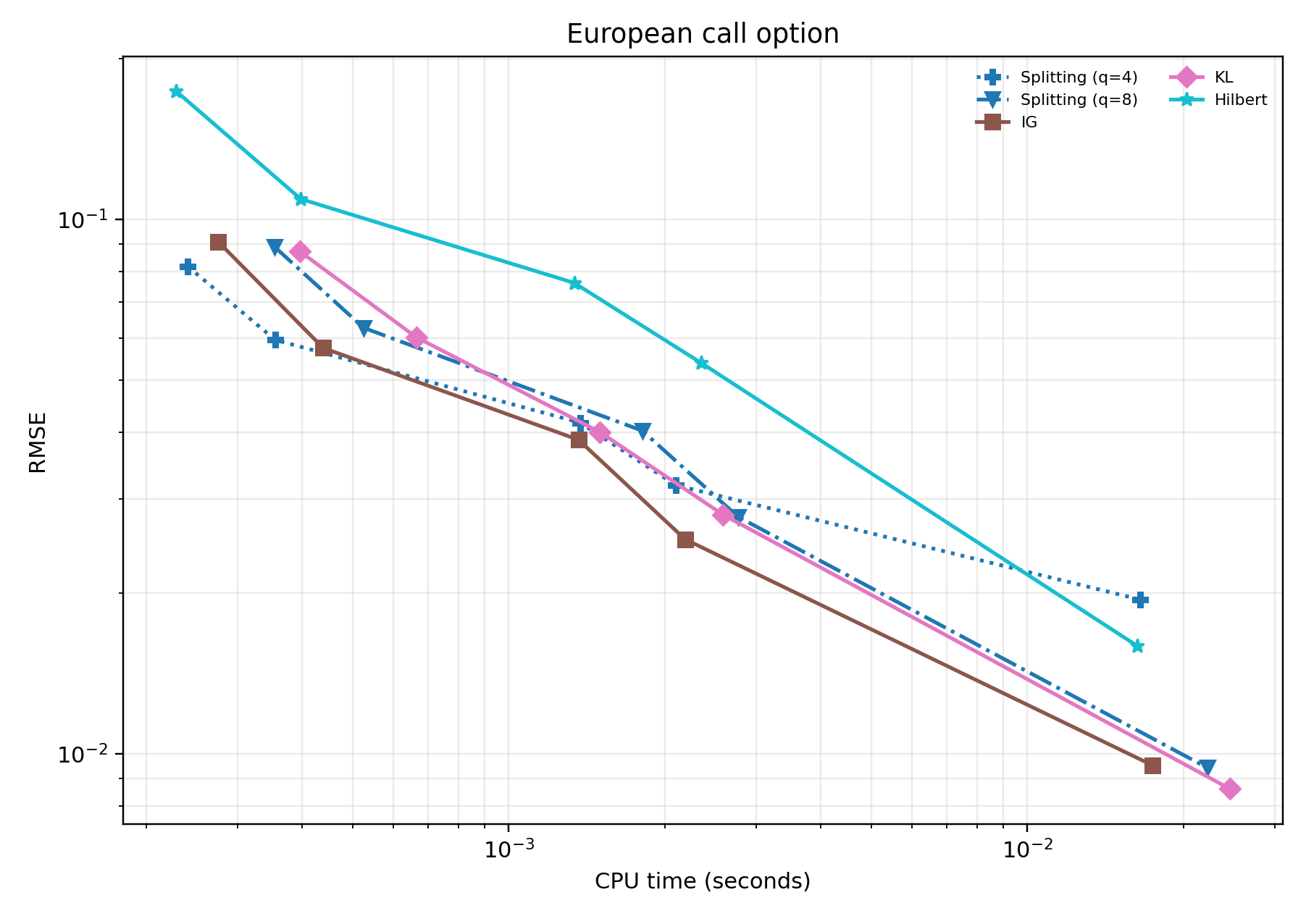}
\else
\includegraphics[width=\ORLfigurewidth]{figures/european_cpu_rmse.png}
\fi
\caption{European vanilla call: plots of RMSE against CPU time for splitting with $q=4$ and
$q=8$, IG, KL, and Hilbert. The 3 simulation schemes: splitting (\(q=8\)), IG and KL show a similar level of performance in accuracy-speed tradeoff.}
\label{fig:european}
\end{figure}

\subsection{Arithmetic Asian call option}\label{sec:asian}
\noindent
The arithmetic Asian call is based on the average of the asset prices observed
over the life of the option.  We use $M=52$ weekly monitoring instants, excluding
the initial price $S_0$, and simulate the discounted Asian call option payoff
\[
e^{-rT}\left(\frac1M\sum_{j=1}^M S_{t_j}-K\right)^+,
\]
where $S_{t_j}$ is the asset price at monitoring instant $t_j, j =1,2, \ldots, M$.

Since the Asian call option payoff depends on the
whole sequence of monitored prices, the conditioning procedure would 
not be the same as that in the vanilla option.
However, we observe that the first independent
$W^{(1)}$ Brownian term is common to all monitored prices.  Since it occurs
before the first monitoring date, it enters every subsequent asset price in the
same manner and can be factored out from the arithmetic average.
For splitting, IG, and KL, we first neglect this Brownian term $W^{(1)}$ and generate
all the remaining variables.  For path $i$, we write the neglected Brownian term as
$\sqrt{V_i}Z_i$, where $Z_i$ is the standard normal variate and $V_i$ is its conditional
variance.  Let $\widetilde S_{t_j,i}$ denote the simulated monitored asset price 
in the $i^{\text{th}}$ path obtained when
the Brownian term is neglected, while all other simulated quantities
remain unchanged.  Since the Brownian term enters in each subsequent asset price multiplicatively, 
we can restore it by multiplying the earlier simulated monitored price $\widetilde S_{t_j,i}$
by the same factor $\exp(\sqrt{V_i}Z_i)$:
\[
S_{t_j,i}=\widetilde S_{t_j,i}\exp(\sqrt{V_i}Z_i),
\qquad j=1,\ldots,M.
\]
We write
\[
A_i^{(0)}=\frac1M\sum_{j=1}^M\widetilde S_{t_j,i},
\qquad
F_{A,i}=A_i^{(0)}\exp\left(\frac{V_i}{2}\right),
\qquad
i= 1,2,\cdots,N.
\]
Here, $A_i^{(0)}$ is the $i^{\text{th}}$ simulated arithmetic average with the neglected Brownian term, 
and $F_{A,i}$ is the corresponding conditional forward.  The restored
simulated arithmetic average $A_i$ then satisfies
\[
A_i=\frac1M\sum_{j=1}^M S_{t_j,i}
=F_{A,i}\exp\left(-\frac{V_i}{2}+\sqrt{V_i}Z_i\right).
\]
Conditional on the other simulated quantities, the arithmetic average is
seen to be lognormal in $Z_i$.  
Let $P_i$ be the conditional Monte Carlo value of the Asian call for path
$i$ after $Z_i$ has been integrated out.  To further reduce sampling variance,
we use $C_i=e^{-rT}F_{A,i}$ as the control variable in
Eq.~\eqref{eq:additive-control}. 
Integrating out with respect to $Z_i$ gives
\begin{equation}
P_i=e^{-rT}\operatorname{BSCall}(F_{A,i},V_i;K).
\label{eq:asian-cmc}
\end{equation} 
Since $F_{A,i}$ is the conditional mean of
the arithmetic average, $C_i$ is the conditional expected discounted
arithmetic average.  It is positively correlated with $P_i$, and its
unconditional mean is known from
$\E[S_{t_j}]=S_0e^{rt_j}$ under the risk-neutral measure:
\[
c_0=\E[C_i]
=e^{-rT}\frac{S_0}{M}\sum_{j=1}^M e^{rt_j}.
\]
We apply Eq.~\eqref{eq:additive-control} to the paired values $(P_i,C_i)$.
This additive adjustment leaves the monitored path unchanged, unlike applying
the European forward rescaling in Eq.~\eqref{eq:multiplicative-correction} to
all monitored prices.  

As a remark, Hilbert does not reveal the common Brownian term
used in the preceding conditioning step.  For this scheme, we may choose $P_i$ to be the direct
simulated Asian call payoff and $C_i$ to be the direct simulated discounted arithmetic
average. The same additive control formula in Eq.~\eqref{eq:additive-control} can be applied.

Figure~\ref{fig:asian} shows the plot of RMSE against CPU time for 
the 4 simulation schemes for pricing arithmetic Asian call. 
The reference true option value is $7.141068705$, which is computed using the
recursion quadrature method of \citet{ZhangZengKwok2023}.
Splitting is seen to be most competitive in accuracy-speed tradeoff among all simulation schemes. 
The other three schemes show a similar level of performance. While the implementation of the Hilbert scheme
is most tedious, it is least competitive in accuracy-speed tradeoff.


\begin{figure}[H]
\centering
\ifdefined\ORLtwocolumn
\includegraphics[width=0.92\ORLfigurewidth]{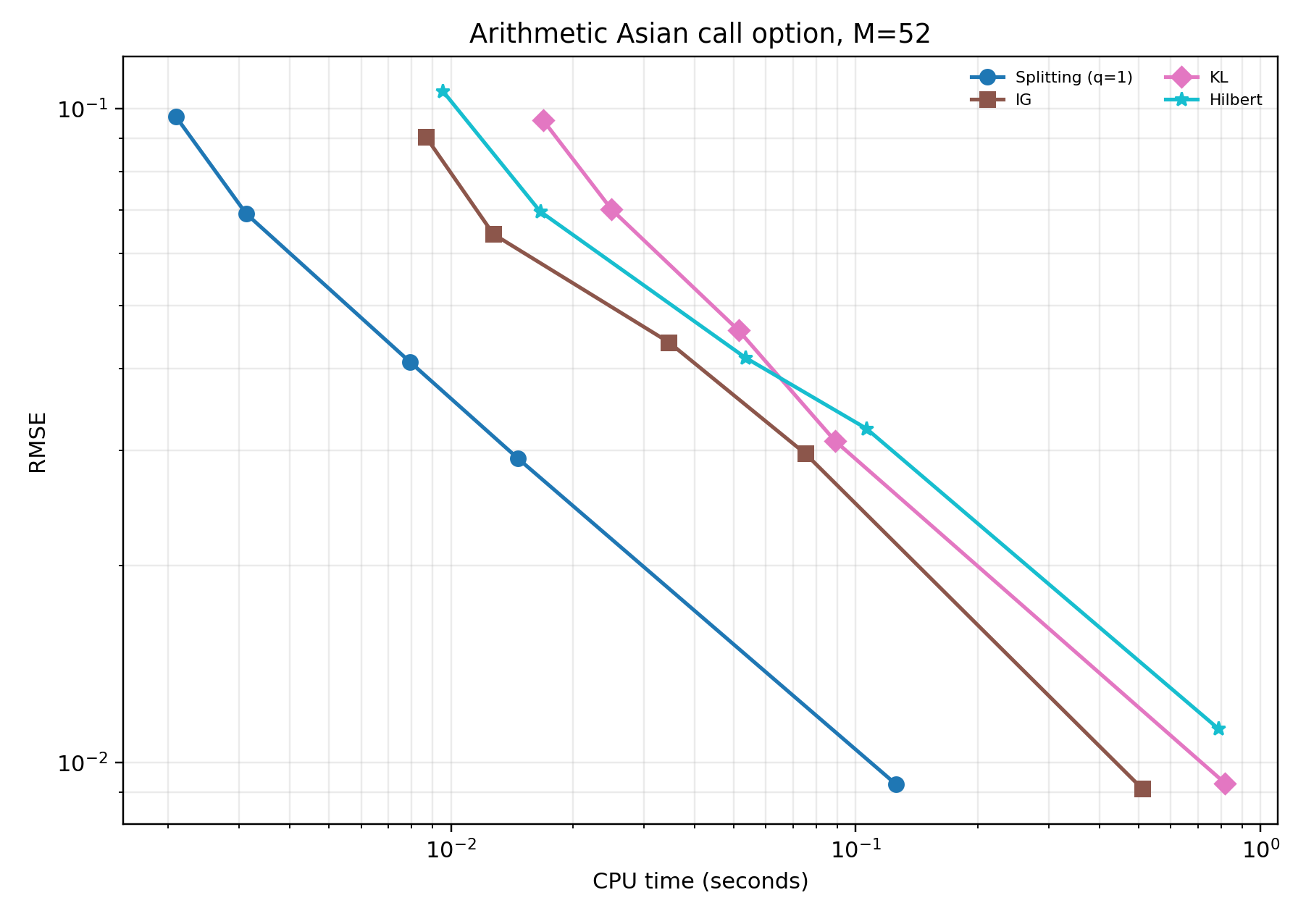}
\else
\includegraphics[width=\ORLfigurewidth]{figures/asian_cpu_rmse.png}
\fi
\caption{Weekly arithmetic Asian call ($M=52$): plots of RMSE against CPU time for 
splitting ($q=1$), IG, KL, and Hilbert schemes. Splitting is seen to be most competitive in
accuracy-speed tradeoff among all simulation schemes. The other three schemes show a similar level of performance.}
\label{fig:asian}
\end{figure}

\subsection{Down-and-in barrier call}\label{sec:barrier}
\noindent
The barrier $B = 90$ is monitored at $M =52$ weekly
instants over the option life of one year. For down-and-in barrier
call option, the discounted barrier call payoff is
$e^{-rT}(S_T-K)^+\mathbf{1}_{\{\tau_B\le T\}}$, 
where $\tau^B$ is the first monitoring time at which
the barrier B is breached.

The following conditioning method works for all four simulation methods. For simulation path \(i\), let \(\tau_i^B\) denote the first monitoring time at which path \(i\) breaches the barrier \(B\); that is,
\(
\tau_i^B := \inf\{t_j:S_{t_j,i}\le B,\; j=1,\ldots,M\}.
\) On the occurrence of the event \(\tau_i^B\le T\), the conditional expectation of the terminal barrier call payoff is the same as the European vanilla call value with asset price \(S_{\tau_i^B}\), volatility \(\sigma_{\tau_i^B}\), and time to maturity \(T-\tau_i^B\).
If \(\tau_i^B=\infty\) (no hitting of barrier on all monitoring dates), the terminal barrier call payoff is set to be zero. Let \(C_{\mathrm{Euro}}(S,\sigma,\tau)\) be the value of the European call with asset price \(S\), volatility \(\sigma\), and time to maturity \(\tau\).
We define \(P_i\) to be the conditional expectation of the terminal barrier call payoff, where
\begin{equation}
P_i=
\begin{cases}
e^{-r\tau_i^B}C_{\mathrm{Euro}}
\left(S_{\tau_i^B},\sigma_{\tau_i^B},T-\tau_i^B\right),
& \text{if }\tau_i^B\le T,\\
0,
& \text{if }\tau_i^B=\infty.
\end{cases}
\end{equation}
We define the control
\(
C_i=
\begin{cases}
e^{-r\tau_i^B}S_{\tau_i^B,i},&\text{if }\tau_i^B\le T,\\
e^{-rT}S_{T,i},&\text{if }\tau_i^B=\infty,
\end{cases}
\) and we have \(\E[C_i]=S_0\), since discounted asset price is a martingale. After the above conditioning procedure, again we apply Eq.~\eqref{eq:additive-control} to the paired values \((P_i,C_i)\).

\Cref{fig:barrier} shows the plots of RMSE against CPU for the four simulation methods. The reference true value of the down-and-in barrier call option is found to be \(1.568940412\), which is obtained from an independent implementation of the conditioning of \citeauthor{Choi2025}'s simulation scheme (\citeyear{Choi2025}) with 5 million paths and \(L=48\). The reported standard error is \(9.62\times10^{-4}\).
The plots reveal that the IG, KL and Hilbert schemes show a similar level of performance at large number of simulation paths, and splitting has the best performance in accuracy-speed tradeoff.

\begin{figure}[H]
\centering
\ifdefined\ORLtwocolumn
\includegraphics[width=0.92\ORLfigurewidth]{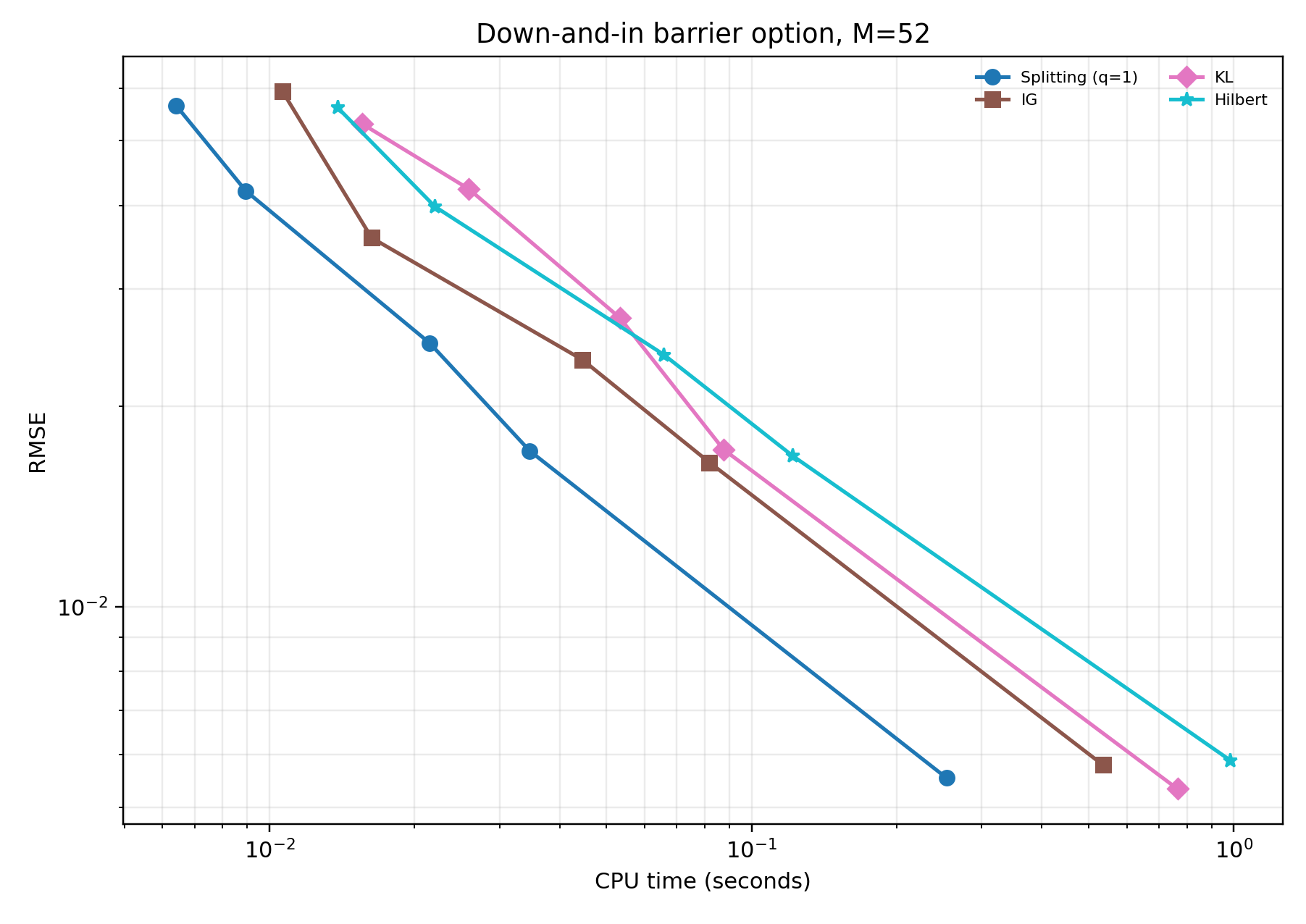}
\else
\includegraphics[width=\ORLfigurewidth]{figures/barrier_cpu_rmse.png}
\fi
\caption{Weekly monitored down-and-in call ($M=52$): plots of RMSE against
CPU time for splitting ($q=1$), IG, KL, and Hilbert.  IG, KL and Hilbert 
show a similar level of performance at larger number of
simulation paths $N$. Splitting is seen to have the 
best performance in accuracy-speed tradeoff among all simulation schemes.}
\label{fig:barrier}
\end{figure}

\subsection{Corridor variance swap}\label{sec:corridor-var-swap}
\noindent
In the corridor variance swap, the floating leg retains a squared log return only when
the asset price at the beginning of the monitoring interval lies in the corridor $(L,U)$,
where $L$ and $U$ are the lower and upper barriers, respectively.
We take $L=80$, $U=120$, and weekly monitoring over the life of one year.
We define the corridor realized variance with $M$ monitoring dates by
\begin{equation}
\operatorname{CRV}_M(L,U)
=
\frac{1}{T}
\sum_{j=1}^{M}
\left(X_{t_j}-X_{t_{j-1}}\right)^2
\mathbf{1}_{\{L\leq S_{t_{j-1}}\leq U\}},
\qquad M=52.
\label{eq:corridor-rv}
\end{equation}
The fair strike of the corridor variance swap is defined by
\[
K_{\mathrm{CRV}}
=
\mathbb{E}\!\left[\operatorname{CRV}_M(L,U)\right],
\]
which can be estimated by simulation by taking the average of
$\operatorname{CRV}_M(L,U)$. Since there is no optionality in the variance swap, the conditioning method is not
applicable. For variance reduction using control variates, for all four simulation methods, we use
\(
C_i=e^{-rT}S_{i,T},\) whose known mean is
\(c_0=\mathbb{E}[C_i]=S_0.
\) The reference true value is $0.034198980391$, which is obtained by integrating the
one-dimensional Fourier representation of \citet{ZhengKwok2014} using the
conditional characteristic function of the OU model.

\Cref{fig:corridor} compares the performance in accuracy--speed tradeoff among the 4 simulation
schemes for finding the fair strike of the corridor variance swap. Again, splitting is seen
to give the best performance in accuracy--speed tradeoff.
\begin{figure}[H]
\centering
\ifdefined\ORLtwocolumn
\includegraphics[width=0.92\ORLfigurewidth]{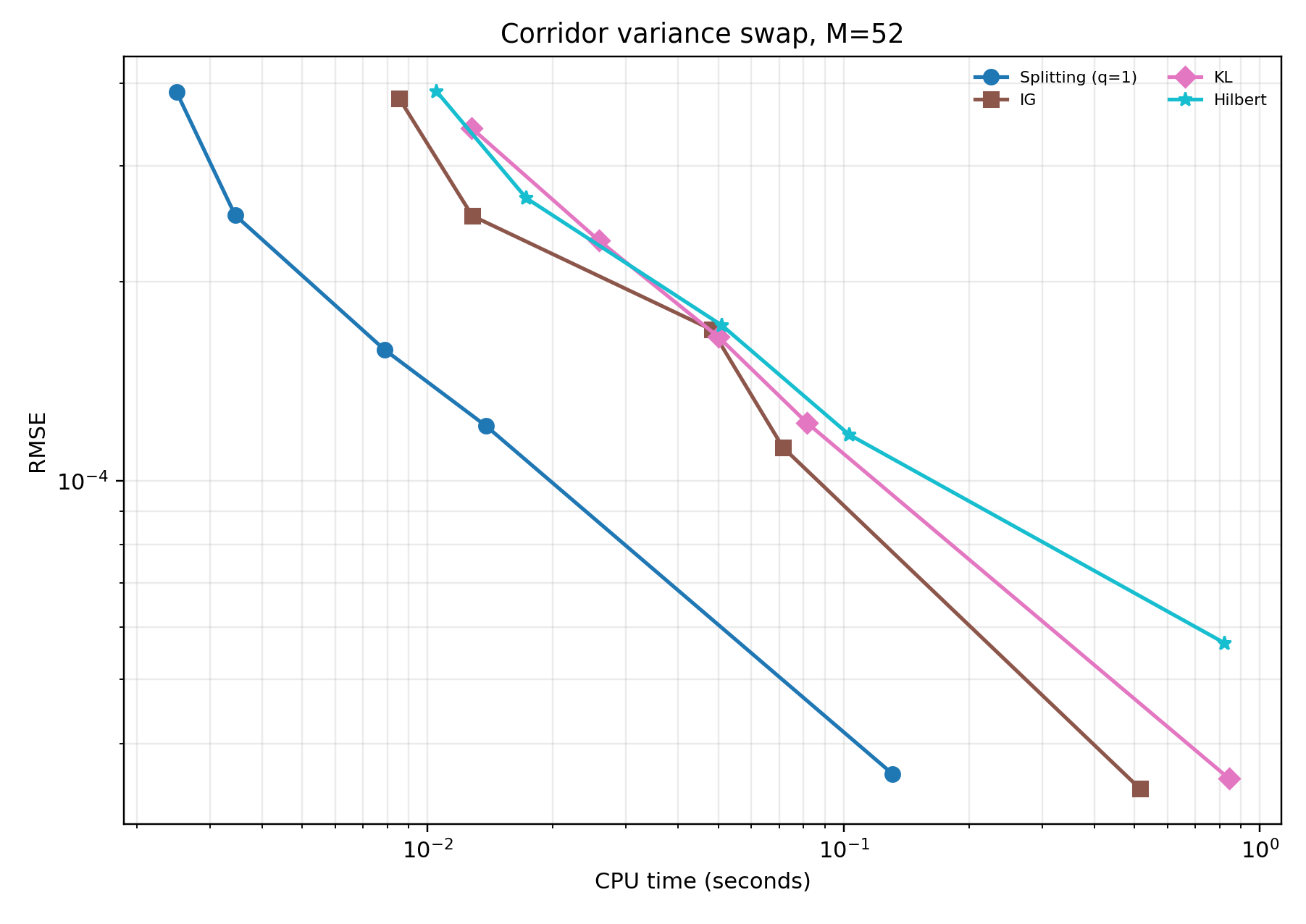}
\else
\includegraphics[width=0.90\linewidth]{figures/corridor_variance_swap_cpu_rmse.png}
\fi
\caption{Weekly corridor variance swap ($M=52$): plots of RMSE against CPU
time for splitting ($q=1$), IG, KL, and Hilbert.  Splitting gives the best
performance in accuracy-speed tradeoff among all simulation schemes.}
\label{fig:corridor}
\end{figure}

\section{Conclusion}\label{sec:conclusion}
\noindent
We develop an efficient Monte Carlo simulation scheme 
based on the operator splitting approach for pricing 
options under the OU driven stochastic volatility model. 
Thanks to the nice analytical tractability in all suboperators 
in the operator splitting procedure, the implementation of 
the scheme only requires simulation of normal variates. 
The numerical challenges of computing conditional 
integrated variance 
and pathwise calculation of inverse integral transform of 
characteristic function are avoided. The computational 
steps are much simplified when compared with the other 
three pioneering simulation schemes proposed in recent 
years. We also include variance reduction via conditioning
and a martingale-preserving control variate as additional 
efficiency enhancement procedures. The above nice features 
that achieve computational efficiency in our proposed 
scheme are well verified by numerical tests of 
accuracy-speed tradeoff comparisons among our scheme 
and three other simulation schemes. The advantages of 
our schemes are best demonstrated in pricing path 
dependent options with larger number of monitoring 
instants. Besides the introduction of a new efficient 
simulation scheme, we also present a summary of existing 
simulation schemes that use different innovative 
approaches that minimize the numerical bottlenecks in 
computing conditional integrated variance and inverse integral 
transform calculations. Interestingly, these other 
schemes also show a similar level of performance in 
accuracy-speed tradeoff but they all lose to our operator 
splitting scheme in pricing path dependent options.    

\section*{Acknowledgment}
\noindent The work of Yue Kuen Kwok was supported by Project 2023CX10X1 and the Guangzhou-HKUST (GZ) Joint Funding Program under Grant Number 2024A03J0630.

\section*{Declaration of Interest Statement}
\noindent There is no interest to declare.

\bibliography{ORL_OUSV_references}

\end{document}